\documentclass[11pt]{llncs}
\usepackage{amsfonts}
\usepackage{amsmath}
\usepackage{complexity}
\usepackage{url}
\usepackage[pdftex]{graphicx}

\usepackage{MnSymbol}

\begin{document}

\title{Deciding the Winner of an Arbitrary Finite Poset Game is PSPACE-Complete}
\author{Daniel Grier\thanks{This work was funded in part by the South Carolina Honors College Science Undergraduate Research Funding Program.  This work was also supported by the Barry M. Goldwater Scholarship.}}
\institute{University of South Carolina \\ \email{grierd@email.sc.edu}}
\date{} 

\maketitle
\begin{abstract}
A poset game is a two-player game played over a partially ordered set (poset) in which the players alternate choosing an element of the poset, removing it and all elements greater than it.  The first player unable to select an element of the poset loses.  Polynomial time algorithms exist for certain restricted classes of poset games, such as the game of Nim.  However, until recently the complexity of arbitrary finite poset games was only known to exist somewhere between $\NC^1$ and $\PSPACE$.   We resolve this discrepancy by showing that deciding the winner of an arbitrary finite poset game is $\PSPACE$-complete.  To this end, we give an explicit reduction from Node Kayles, a $\PSPACE$-complete game in which players vie to chose an independent set in a graph.
\end{abstract}
 \section{Introduction}
A partially ordered set, or poset, is a set of elements with a binary relation (denoted $\le$) indicating the ordering of elements that is reflexive, transitive, and antisymmetric.  A poset game is an impartial two-player game played over some poset.  Each turn, a player selects an element of the poset, removing it and all elements greater than it.  A player loses when faced with the empty set.  Equivalently, the last player able to select an element wins.  We will assume that the number of elements in the poset is finite, which ensures that the game will eventually end in such a manner.

Poset games have been studied in various forms since a complete analysis of the game of Nim was given in 1901 by C. Bouton \cite{bouton:1901}.  Other poset games with explicit polynomial time strategies include Von Neumann's Hackendot \cite{ulehla:1980} and impartial Hackenbush on trees \cite{berlekamp:2004}.  The above games have no induced subposet of cardinality four that form an `N'.  In fact, it is shown in \cite{deuber:1996} that all N-free poset games can be solved in polynomial time.  

However there are several other well-studied poset games played over specific structures with unknown complexity \cite{fraenkel:2000}.  Perhaps the most popular is the game of Chomp, which was introduced by Gale in 1974 and is played on the cross product of two Nim stacks \cite{gale:1974}.  Work by Byrnes \cite{byrnes:2003} shows that certain Chomp positions exhibit periodic behavior, but a quick general solution still does not exist.  In Subset Takeaway \cite{gale:1982}, introduced by Gale in 1982, the players take turns removing a set and all its supersets from a collection of sets.  In Shuh's Game of Divisors \cite{schuh:1952}, the players alternate removing a divisor of $n$ and its multiples.  In fact, both Chomp and Subset Takeaway are special cases of the Game of Divisors, with $n$ the product of at most two primes and $n$ square-free, respectively.

In this paper, we discuss the complexity of deciding the winner of an arbitrary finite poset game, which has remained a longstanding question in the attempt to classify the tractability of combinatorial games \cite{fraenkel:2000,fraenkel:2004}.  Let PG be the language consisting of poset games with a winning strategy for the first player.  In \cite{kalinich:2011}, Kalinich shows that PG is at least as hard as $\NC^1$ under $\AC^0$ reductions by creating a correspondence with boolean circuits.  Weighted poset games, which are a generalization of poset games, were shown to be $\PSPACE$-complete in \cite{ito:2011}.  That result, which uses a completely different technique than the one described in this paper, along with another proof in \cite{soltys:2011}, clearly show that PG is in $\PSPACE$.  We show that PG is indeed \PSPACE-complete.
 
In \cite{schaefer:1978}, Schaefer shows that the two-player game Node Kayles is $\PSPACE$-complete.  In Node Kayles, the players take turns removing a vertex and all neighbors of that vertex from a graph.  The first player unable to move loses.  

In Section~\ref{section:constructions} we will give two constructions that serve as the basis for a reduction from Node Kayles to PG.  We will then give a variety of lemmas demonstrating the desirable properties of these constructions in Section~\ref{section:lemmas}.  In Section~\ref{section:main_theorem} we will combine these lemmas to show that PG is $\PSPACE$-complete.

\section{Constructions}
\label{section:constructions}
Below we will give two constructions, $\psi$ and $\varphi$.  When applied in succession, they reduce an instance of Node Kayles into an instance of PG such that the winning player is preserved.  Let G be the class of finite simple graphs and $P$ be the class of finite posets.  For $g \in G$ we will write $g = (V, E)$ where $V$ is the set of vertices and $E$ is the set of edges.  We will use $K_n$ to denote the complete graph on $n$ vertices.
\subsection{$\psi$-construction}
Define $\psi : G  \rightarrow G$ such that
\begin{itemize}
\item $\left|E\right|$ is odd $\implies \psi(g) = g \cupdot K_2 \cupdot K_2$
\item $\left|E\right|$ is even $\implies \psi(g) = g \cupdot K_2 \cupdot K_4$
\end{itemize}

\begin{figure}[!h]
\centering
\includegraphics[scale=.9]{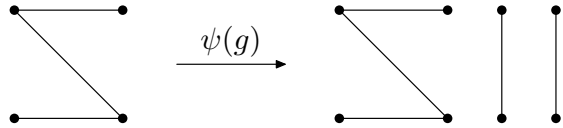}
\caption{Example of $\psi$-construction when $\left|E\right|$ is odd.}
\end{figure}

\begin{figure}[!h]
\centering
\includegraphics[scale=.9]{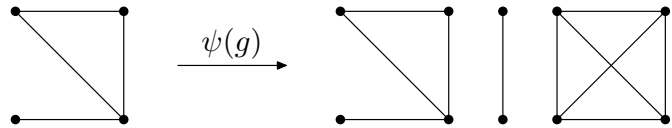}
\caption{Example of $\psi$-construction when $\left|E\right|$ is even.}
\end{figure}

This construction serves two purposes.  First, the edge cardinality of the resulting graph is always odd.  Second, for every vertex, there is an edge that is not incident to it.  It is also important to note that the winning player of the Node Kayles game does not change (see Lemma~\ref{psi_lemma}).

\subsection{$\varphi$-construction}
Let $\varphi : G \rightarrow P$ be a function from simple graphs to posets, where $\varphi(g) = A \cup B \cup C$ is a three-level poset with disjoint levels $A$, $B$, and $C$ from lowest to highest.  That is, for any $a \in A$, $b \in B$, and $c \in C$, $b \not\leq a$, $c \not\leq b$, and $c \not\leq a$.  Furthermore, any two elements on the same level are incomparable.  \\ \\
Fix $g = (V, E)$.  The elements of the poset $\varphi(g)$ are as follows:  
\begin{itemize} \itemsep 0pt
\item The elements of $C$ are the edges of $g$.  That is, $C = E$.
\item The elements of $B$ are the vertices of $g$.  That is, $B = V$.
\item The elements of $A$ are copies of the edges of $g$. To represent this, let $\gamma : C \rightarrow A$ be a 1-1 correspondence between the elements of $C$ and the elements of $A$.   
\end{itemize}
For each edge $e = (v_1, v_2) $ and $b \in B$, the $\le$ relationship of the poset $\varphi(g)$ is as follows:  
\begin{itemize} \itemsep 0pt
\item $b \le e$ iff $b = v_1$ or $b = v_2$.  That is, $e$ lies directly above its endpoints in $B$.
\item $\gamma(e) \le b$ iff $b \neq v_1$ and $b \neq v_2$.  That is, $\gamma(e)$ is less than all the elements in $B$ except the endpoints of $e$.
\end{itemize}

\begin{figure}[!h]
\centering
\includegraphics{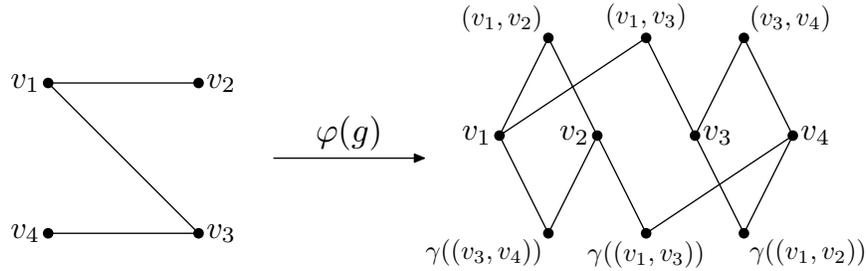}
\caption{Example $\varphi$-construction.  Note that the left picture is an undirected graph representing a Node Kayles game, and the right picture is a Hasse Diagram representing the resultant poset game.}
\end{figure}

\section{Lemmas}
\label{section:lemmas}
\begin{lemma}
\label{psi_lemma}
Player 1 wins the Node Kayles Game on $g$ iff Player 1 wins the Node Kayles Game on $\psi(g)$.
\end{lemma}
\begin{proof}
Suppose that the Node Kayles game played on $g$ is a win for Player 1, who we will assume by convention is the first to play.  We will show that this gives Player 1 an explicit winning strategy on $\psi(g)$.  Player 1 first chooses the winning move in $g$.  If Player 2 chooses a vertex in $g$, Player 1 can always respond with another move in $g$ because Player 1 has the winning strategy on $g$.  If Player 2 chooses a vertex in one of the complete graphs, Player 1 can respond with a vertex in the other complete graph, removing both complete graphs from consideration for the remainder of the game.  Because Player 1 can respond to any move of Player 2, Player 1 will eventually win.  Of course, this argument holds if Player 2 has the winning strategy in $g$, and similarly shows that a player has a winning strategy on $g$ if he has a winning strategy on $\psi(g)$.

In terms of Sprague-Grundy theory, the disjoint union of the two complete graphs has Grundy number zero.  Adding a game of Grundy number zero to an existing game does not change the winner of the original game \cite{berlekamp:2004}.   In particular, the Grundy number of $g$ is equal to the Grundy number of $\psi(g)$.  \qed
\end{proof}

Let $g = (V, E)$ be a finite simple graph and $e = (v_1, v_2)$ be an arbitrary edge in $\psi(g)$.  For the following lemmas, assume that two players are playing the poset game on $\varphi(\psi(g))$.  Also assume, for simplicity, that the players are Alice and Bob.
\begin{lemma}
\label{successful_challenge_lemma}
Assume no moves in $A$ or $C$ have yet been chosen.  If both $v_1$ and $v_2$ have been chosen, then $\gamma(e)$ is a winning move.
\end{lemma}
\begin{proof}
Because the $\psi$-construction always leaves a graph with an odd number of edges, choosing $\gamma(e)$ leaves an even number of incomparable points in $A$.  \qed
\end{proof}

\begin{lemma}
\label{bad_challenge_lemma}
Assume no moves in $A$ or $C$ have yet been chosen.  If exactly one of $v_1$ and $v_2$ has been chosen, then $\gamma(e)$ is a losing move.
\end{lemma}
\begin{proof}
First notice that $e$ has already been removed from the poset because both $v_1 \le e$ and $v_2 \le e$.  Because $\gamma(e) \not\leq v_1$ and $\gamma(e) \not\leq v_2$, choosing $\gamma(e)$ leaves a single point (either $v_1$ or $v_2$) in $B$. Thus, the next player can win by choosing the lone element in $B$, leaving an even number of incomparable points in $A$.  \qed
\end{proof}

\begin{lemma}
\label{worse_challenge_lemma}
Assume no moves in $A$ or $C$ have yet been chosen.  If neither $v_1$ nor $v_2$ has been chosen, then both $e$ and $\gamma(e)$ are losing moves.
\end{lemma}
\begin{proof}
Assume that either player, say Alice, chooses $\gamma(e)$, which results in an even number of incomparable points in $A$, $v_1$ and $v_2$ in $B$, and $e$ in $C$.  Bob can then respond by choosing $e$.  If Alice responds with $v_1$, then Bob can respond with $v_2$ (and vice versa), resulting in an even number of points in $A$, which is a win for Bob.  

If, however, Alice responds with a point $a \in A$, there are three cases: $a \le v_1$ and $a \le v_2$, $a \le v_1$ and $a 
\not\leq v_2$, or $a \le v_2$ and $a \not\leq v_1$.  Note that, by construction, there is no point $a$ such that $a \not\leq v_1$ and $a \not\leq v_2$.  That is, the only point that is not less than both $v_1$ and $v_2$ is $\gamma(e)$, which has already been taken by assumption.  So first assume that $a \le v_1$ and $a \le v_2$.  This would leave an odd number of elements in $A$, resulting in a win for Bob.  Consider then that $a \le v_1$ and $a \not\leq v_2$ or $a \le v_2$ and $a \not\leq v_1$.  Without loss of generality we can assume $a \le v_1$ and $a \not\leq v_2$.  Because $\psi(g)$ has at least two distinct components, each having at least one edge, there exists an edge $e_2$ that is not incident to either $v_1$ or $v_2$.  By construction, $\gamma(e_2) \le v_2$.  Thus, Bob can choose $\gamma(e_2)$, leaving only an even number of elements in $A$, resulting in a win for Bob.  

If Alice had initially chosen $e$ instead of $\gamma(e)$, then Bob could have responded with $\gamma(e)$, which leads to the same game as played as above, which was a win for Bob.  \qed
\end{proof}

\section{Main Theorem}
\label{section:main_theorem}
\begin{theorem}
\label{main_theorem}
PG is $\PSPACE$-complete.
\end{theorem}
\begin{proof}
It is straightforward to check and demonstrated explicitly in \cite{soltys:2011} that PG is in $\PSPACE$.  We will next give a reduction from Node Kayles to PG to show that the latter is also $\PSPACE$-hard.  First note that $\varphi(\psi(g))$ is computable in polynomial time.

We will argue inductively that Player 1 has a winning strategy for the poset game played on $\varphi(\psi(g))$ iff Player 1 has a winning strategy for the Nodes Kales game played on $g$.  The idea behind the construction is that both players are forced to play elements in $B$ until two elements $v_1$ and $v_2$ representing adjacent vertices in $\psi(g)$ have been chosen.Ê At this point the following player can win by choosing the element $\gamma((v_1,v_2))$ in $A$.  

Assume that the poset game played on $\varphi(\psi(g))$ has been played in the prescribed manner so far.  That is, no elements from $A$ or $C$ have yet been chosen.  Lemma~\ref{successful_challenge_lemma} ensures that choosing a vertex neighboring a vertex that has already been chosen is a losing move.  Lemma~\ref{bad_challenge_lemma} and Lemma~\ref{worse_challenge_lemma} ensure that choosing any point in $A$ or $C$ before two neighboring vertices have been chosen is a losing move.  Thus, a player has a winning strategy on $\varphi(\psi(g))$ iff that player has a winning strategy on $\psi(g)$, since there is an obvious correspondence between the moves in $\varphi(\psi(g))$ and the moves in $\psi(g)$.  Lemma~\ref{psi_lemma} ensures that a player has a winning strategy on $\psi(g)$ iff he has a winning strategy on $g$.  \qed
\end{proof}

\begin{figure}[!h]
\centering
\includegraphics{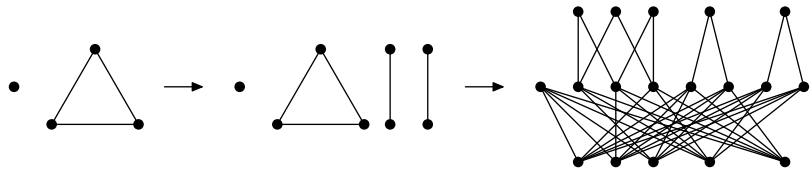}
\caption{Example of full reduction from $g$ to $\psi(g)$ to $\varphi(\psi(g))$}
\end{figure}

\section{Future Work}
Using the above theorem, it follows easily that deciding the winner of a finite poset game with any height $k \ge 3$ is \PSPACE-complete.  In contrast, determining the winner of single-level poset games is trivially obtained by considering the parity of the poset elements.  There are also polynomial time algorithms for some two-level poset games.  In \cite{fraenkel:1991}, Fraenkel and Aviezri give a polynomial time algorithm for finding the Grundy number of poset games played over a restricted class of two-level posets whose upper elements act like edges of a hypergraph.  In \cite{fenner:2012}, Fenner,  Gurjar,  Korwar, and Thierauf give a natural generalization of that algorithm and explore other possible avenues for finding the winner in polynomial time.  However, neither of these results yield a general algorithm, and the complexity of two-level poset games remains an open problem.

This work has also spawned a new \PSPACE-complete game on sets invented by Fenner and Fortnow \cite{cc_blog:2012}.  Given a collection of finite sets $S_1, \ldots, S_k$, each player takes turns picking a non-empty set $S_i$, removing the elements of $S_i$ from all the sets $S_j$. The player who empties all the sets wins.  To reduce a poset game into an instance of set-game, simply take the sets as the upper cones of the poset.  That is, each set consists of an element and all elements greater than it.  However, if the cardinality of the sets is bounded, the complexity is still open.

\section*{\large{Acknowledgements}}
I would like to thank Dr. Stephen Fenner for almost everything leading to this result.  Perpetually busy, he still always finds the time to teach me and listen to my ideas.  I am also very grateful for the support I received from the University of South Carolina Honors College and for all of those who helped me edit and refine this paper. 

\bibliographystyle{plain}
\bibliography{poset_bib}

\begin{thebibliography}{10}

\bibitem{berlekamp:2004}
E.R. Berlekamp, J.H. Conway, and R.K. Guy.
\newblock {\em Winning {W}ays for {Y}our {M}athematical {P}lays. {V}olume 1}.
\newblock AK Peters, 2004.

\bibitem{bouton:1901}
C.L. Bouton.
\newblock Nim, a game with a complete mathematical theory.
\newblock {\em The Annals of Mathematics}, 3(1/4):35--39, 1901.

\bibitem{byrnes:2003}
S.~Byrnes.
\newblock Poset game periodicity.
\newblock {\em Integers: Electronic Journal of Combinatorial Number Theory},
  3(G03):2, 2003.

\bibitem{deuber:1996}
W.~Deuber and S.~Thomass{\'e}.
\newblock Grundy sets of partial orders.
\newblock 1996.

\bibitem{fenner:2012}
S.~Fenner, R.~Gurjar, A.~Korwar, and T.~Thierauf.
\newblock Two-level posets.
\newblock {\em manuscript}, 2012.

\bibitem{cc_blog:2012}
Lance Fortnow.
\newblock A simple {PSPACE}-complete problem. \\
  \url{http://blog.computationalcomplexity.org/2012/11/a-simple-pspace-complete-problem.html}.

\bibitem{fraenkel:1991}
A.~S. Fraenkel and E.~R. Scheinerman.
\newblock A deletion game on hypergraphs.
\newblock {\em Discrete Applied Mathematics}, 30(2-3):155--162, 1991.

\bibitem{fraenkel:2000}
Aviezri~S Fraenkel.
\newblock {R}ecent results and questions in combinatorial game complexities.
\newblock {\em Theoretical computer science}, 249(2):265--288, 2000.

\bibitem{fraenkel:2004}
Aviezri~S Fraenkel.
\newblock {C}omplexity, appeal and challenges of combinatorial {G}ames.
\newblock {\em Theoretical Computer Science}, 313(3):393--415, 2004.

\bibitem{gale:1974}
D.~Gale.
\newblock A curious {N}im-type game.
\newblock {\em The American Mathematical Monthly}, 81(8):876--879, 1974.

\bibitem{gale:1982}
D.~Gale and A.~Neyman.
\newblock Nim-type games.
\newblock {\em International Journal of Game Theory}, 11(1):17--20, 1982.

\bibitem{ito:2011}
H.~Ito and S.~Takata.
\newblock {PSPACE}-completeness of the weighted poset game.
\newblock 2011.

\bibitem{kalinich:2011}
A.O. Kalinich.
\newblock Flipping the winner of a poset game.
\newblock {\em Information Processing Letters}, 2011.

\bibitem{schaefer:1978}
T.J. Schaefer.
\newblock On the complexity of some two-person perfect-information games.
\newblock {\em Journal of Computer and System Sciences}, 1978.

\bibitem{schuh:1952}
F.~Schuh.
\newblock Spel van delers ({T}he game of divisors).
\newblock {\em Nieuw Tijdschrift voor Wiskunde}, 39:299--304, 1952.

\bibitem{soltys:2011}
M.~Soltys and C.~Wilson.
\newblock On the complexity of computing winning strategies for finite poset
  games.
\newblock {\em Theory of Computing Systems}, 48(3):680--692, 2011.

\bibitem{ulehla:1980}
J.~\'{U}lehla.
\newblock A complete analysis of {V}on {N}eumann's {H}ackendot.
\newblock {\em International Journal of Game Theory}, 9(2):107--113, 1980.

\end{thebibliography}
\end{document}